\RequirePackage{fix-cm}
\documentclass[smallextended]{svjour3}       
\smartqed  
\usepackage[ruled,linesnumbered, noline]{algorithm2e}
\usepackage{lineno,hyperref}
\usepackage[misc]{ifsym}
\usepackage{footmisc}
\usepackage{appendix}

\usepackage{amsmath}
\usepackage{times}
\usepackage{amsfonts}
\usepackage{fancyhdr}
\usepackage{graphicx}
\usepackage{lineno}
\usepackage{array}
\usepackage{longtable}

\usepackage{multirow}
\usepackage{amsthm}

\usepackage[dvipsnames]{xcolor}

\usepackage{ragged2e}

\def \x{\textbf{x}}
\def \z{\textbf{z}}
\def \o{\textbf{o}}
\def \s{\textbf{s}}
\def \t{\textbf{t}}
\def \0{\textbf{0}}
\def \1{\textbf{1}}

\def \y{\textbf{y}}

\def \B{\textbf{B}}

\def \opt{{\mathsf{opt}}}

\def\1{{\mathbf 1}}

\def\DSMS{\mathsf{DrSMC}}

\def\SMK{\mathsf{SMK}}
\def\DRSC{\mathsf{DRSC}}
\def\algone{FastDrSub}
\def\algtwo{FastDrSub+}
\newcommand{\cp}[1]{{\textcolor{blue}{#1}}}

\begin{document}
	
	\title{Fast Approximation Algorithm for Non-Monotone DR-submodular Maximization under Size Constraint
	}
	
	\titlerunning{Fast Approximation Algorithms for Non-Monotone DR-submodular...}        
	
	\author{Tan D. Tran         \and
		Canh V. Pham (\Letter) 
	}
	
	
	\institute{Tan D. Tran \at
		ORLab, Phenikaa School of Computing, Phenikaa University, Hanoi, 12116, Vietnam\\
		\email{tan.trandinh@phenikaa-uni.edu.vn}           
		\and
		Canh V. Pham (Corresponding author) \at
		ORLab, Phenikaa School of Computing, Phenikaa University, Hanoi, 12116, Vietnam\\
		\email{canh.phamvan@phenikaa-uni.edu.vn}
	}
	
	\date{Received: date / Accepted: date}

	\maketitle
	
	\begin{abstract}
		This work studies the non-monotone DR-submodular Maximization over a ground set of $n$ subject to a size constraint $k$. We propose two approximation algorithms for solving this problem named  \algone \ and  \algtwo.  \algone\ offers an approximation ratio of $0.044$ with query complexity of $O(n \log(k))$. The second one, \algtwo\, improves upon it with a ratio of $1/4-\epsilon$ within query complexity of $(n \log k)$ for an input parameter $\epsilon >0$. Therefore, our proposed algorithms are the first constant-ratio approximation algorithms for the problem with the low complexity of $O(n \log(k))$.
		Additionally, both algorithms are experimentally evaluated and compared against existing state-of-the-art methods, demonstrating their effectiveness in solving the Revenue Maximization problem with DR-submodular  objective function. The experimental results show that our proposed algorithms significantly outperform existing approaches in terms of both query complexity and solution quality.
		
		\keywords{Approximation algorithm \and Submodular \and DR-submodular \and Integer Lattice \and Size Constraint }
	\end{abstract}
	
	\section{Introduction}
	Submodular optimization problems have emerged in a wide range of applications, particularly in machine learning \cite{bach2013learning,Das,prajapat2024submodular}, data mining \cite{ap-recommend,ap-datasum}, and combinatorial optimization \cite{vondrak08_welfare,canh-joco19,tap-17}. Numerous applications related to influence propagation in social networks \cite{tap-17,vic-icml-2019,canh-joco19,canh_optlet,canh_csonet19,canh-csonet18}, budget allocation \cite{soma-nips-15}, recommendation systems \cite{recomend-icml}, data summarization \cite{sc-dis-alg-nips15}, and active set selection \cite{Norouzi_NIPS2016} can all be framed as Submodular Maximization problems. 
	A set function $f: 2^E \mapsto \mathbb{R}_+ $, defined on all subsets of a ground set $E$ size $n$,  is submodular iff it satisfies \textit{the diminishing return property}, i.e., for $X \subseteq Y\subseteq E$, and an element $e\notin Y$, we have:
	\begin{align*}
	f(X \cup \{e\}) - f(X) \geq f(Y \cup \{e\})- f(Y).
	\end{align*}
	The Submodular Maximization problem has been studied under various constraints, including unconstrained \cite{feige2011maximizing}, size constraint~\cite{nemhauser1978analysis}, knapsack constraint~ \cite{sviridenko2004note}, matroid constraint~\cite{calinescu2007maximizing}, etc.
	Nevertheless, the traditional submodular set function falls short in addressing specific real-world scenarios that permit multiple instances of an element from the ground set to be selected. To address this limitation, Soma and Yoshida~ \cite{soma-nips-15} proposed an extension of the function $f$ to the integer lattice $\mathbb{Z}^E_+$, introducing a generalized form of submodularity in this context, termed \textit{diminishing return submodular (DR-submodular)}.  For a vector $\x \in \mathbb{Z}^E_+$, denote by $\x(e)$ the value of $\x$’s coordinate corresponding to an element $e$. For $\x, \y \in \mathbb{Z}_+^E$, we say that $\x\leq \y$ iff $\x(e)\leq \y(e), \forall e \in E$.
	The function $f: \mathbb{Z}^E_+\mapsto \mathbb{R}_+ $ is diminishing return submodular (DR-submodular) on integer lattice if:
	$$
	f(\x+\delta\1_e)-f(\x) \geq f(\y+\delta \1_e)-f(\y)
	$$
	for $\x, \y \in \mathbb{Z}^E_+, \x \leq \y$. In  this work, we consider the DR-submodular Maximization under Size constraint ($\DSMS$) problem, defined as follows:
	\begin{definition}[$\DSMS$ problem]
		Given a  DR-submodular function $f:\mathbb{Z}^E_+\mapsto \mathbb{R}_+$, a bounded lattice $\B$, and a positive integer $k>0$. The problem asks to find $\x \leq \B$ with the total size $\|\x\|_1 \leq k$  such that $f(\x)$ is maximized, i.e,
		\begin{align}
		\max_{\x \in \mathbb{Z}^E_+ } \ \ \ f(\x)  \  \ \mbox{subject to:} \ \  \|\x\|_1 \leq k, \0\leq\x \leq \B
		\end{align}
	\end{definition}
	where $\B=B \cdot\1 $ and $\|\x\|_1=\sum_{e \in E}\x(e)$. Without loss of generality, we set $\B=k \cdot\1 $. 
	The integer lattice can be represented as a multiset of size $O(nk)$, allowing the direct adaptation of existing algorithms for the Submodular Maximization under a size constraint ($\DSMS$) problem. However, the fastest known algorithm for $\DSMS$, proposed by~\cite{buchbinder2014submodular}, requires $\Omega(nk)$ oracle queries to evaluate the function $f$, which is not polynomial in the size of the input. Importantly, since $k$ is an integer and can be encoded using only $O(\log k)$ bits, any algorithm with complexity $O(\mathrm{poly}(k))$ is not necessarily polynomial in the input size. In this setting, we assume access to a value oracle that returns $f(\x)$ when queried with a multiset $\x$.
	
	Recently, Ene and Nguyen~\cite{ene2016reduction} proposed a reduction technique that transforms the problem of optimizing DR-submodular functions into a standard submodular maximization problem, thereby enabling the application of well-established results in submodular optimization. Specifically, the $\DSMS$ problem, after applying the reduction~\cite{ene2016reduction}, can be transformed into the following submodular maximization problem under a knapsack constraint:
	\begin{align}
	\max_{S \subseteq E'} \quad g(S) \quad \text{subject to:} \quad c(S) \leq k,
	\end{align}
	where $E' = \bigcup_{e \in E} \{(e,j) : j \in [t_e]\}$ is the new ground set obtained by decomposing each integer variable $\x(e)$ into a sum of smaller weights $a_{e,j} \leq \varepsilon k$. Each $a_{e,j}$ represents a small positive integer weight used to decompose the integer variable $\x(e)$ into a sum of binary variables, ensuring that any value in the range $[0, B_e]$ can be expressed as a subset sum of $\{a_{e,j}\}_{j=1}^{t_e}$. The submodular function $g(S)$ is defined as $g(S) := f(\x)$, where $\x(e) = \sum_{j : (e,j) \in S} a_{e,j}$ for all $e \in E$, and the cost function is $c(S) := \sum_{(e,j) \in S} a_{e,j}$. 
	As a result of the reduction, the size of the new search space becomes $O(n \log k)$ instead of $O(n)$ as in the original DR-submodular problem.   One can apply the best-known approximation ratio, together with efficient practical algorithms for the submodular knapsack problem, as established in~\cite{pham-ijcai23,Han2021_knap}, to obtain a ratio of
	$1/4-\epsilon$.
	However, a weakness	randomized algorithms is that they only keep the ratio in expectation and the experimental results for real-world dataset may be unstable.
	In contrast, recent studies have demonstrated that deterministic algorithms often yield superior empirical performance, primarily due to their stability across diverse datasets~\cite{kuhnle2021bquick,kdd-ChenK23,han-neurips20,li-linear-knap-nip22}. Moreover, as data scales to massive sizes, it becomes increasingly important to design approximation algorithms for $\DRSC$ that minimize oracle queries to ensure computational efficiency.
	
	These observations raise two important research questions: \textbf{(1)} Can we design new algorithms that achieve a comparable approximation ratio while improving the oracle query complexity, thereby enhancing the practical efficiency of the reduction framework? \textbf{(2)} Is it possible to develop a \textit{deterministic} algorithm for the DR-submodular maximization problem under a size constraint with an approximation guarantee comparable to the current randomized methods?
	Addressing these questions would not only improve the computational performance of DR-submodular optimization on large-scale instances but also advance the theoretical understanding of submodular maximization over the integer lattice.
	\paragraph{Our Contributions and Techniques.}
	In this work, we try to address the above two questions. Our main contributions are as follows:
	\begin{itemize}
		\item We propose two algorithms: \algone \ achieves an approximation ratio of $0.044$ with a query complexity of $O(n \log(k))$, and  \algtwo\ achieves an approximation ratio of $\frac{1}{4} - \epsilon$ with a query complexity of $O\left(\frac{n}{\epsilon}\log\left( \frac{k}{\epsilon}\right)\right)$ for the $\DSMS$ problem.
      Although the best-known theoretical approximation ratio is $0.401$ \cite{buchbinder2024constrained}, the corresponding algorithm has polynomial complexity and is therefore impractical. Compared with the fastest near-linear algorithms \cite{pham-ijcai23,Han2021_knap}, which rely on randomness, our approach achieves a similar complexity but remains \emph{deterministic}. To the best of our knowledge, this is the first deterministic algorithm that achieves the strongest known approximation ratio in the near-linear time setting (for details, see Table~\ref{tab:compare}).
        
		\item We perform comprehensive experiments on Revenue Maximization benchmarks, demonstrating that our algorithms attain solution quality comparable to the current state of the art while requiring fewer oracle queries.
	\end{itemize}
    \begin{table*}[h]
		\centering
			\caption{
			We compare several algorithms for the maximization of a non-monotone DR-submodular function under a size constraint. Each algorithm is evaluated by its approximation guarantee and the number of oracle queries required. Let $\epsilon>0$ denote the accuracy parameter, $k$ the size budget, and $\|\mathbf{b}\|_\infty = k$ the maximum coordinate bound. In particular, the oracle complexities of RLA and SMKRANACC have been rederived using the reduction framework of Ene and Nguyen~\cite{ene2016reduction}, which increases the ground‐set size to $n(2\log k + 1)$. Note that the rederived complexities are stated with respect to this enlarged ground set, but they do not explicitly account for the additional computational overhead introduced by the reduction itself.
		}
		\footnotesize{
			\begin{tabular}{cccc}
				\hline
				\textbf{Reference} & \textbf{Approx. Factor} & \textbf{Query complexity}&
				\textbf{Type}\\
				\hline
				RLA \cite{pham-ijcai23}+Reduction\cite{ene2016reduction}       & $\frac{1}{4} -\epsilon$                                   & $O\left(\frac{n}{\epsilon} \log (k) \log ( \frac{1}{\epsilon}) \right)$& Randomized                                                          \\				
				SMKRANACC\cite{Han2021_knap}+Reduction\cite{ene2016reduction}              & $\frac{1}{4}  -\epsilon$                                 & $O\left( \frac{n}{\epsilon} \log (k) \log ( \frac{k}{\epsilon}) \right)$& Randomized                              \\					
				\hline
				\algone \  (this paper)                    & $\displaystyle 0.044$ & \textbf{$O(n\log k)$}& Deterministic                                                       \\
				\algtwo  \ (this paper)                    & $\tfrac14-\epsilon$                          & $O(\frac{n}{\epsilon}\log(\frac{1}{ \epsilon})\log(k))$&  Deterministic         \\
				\hline
		\end{tabular}}
		\label{tab:compare}
	\end{table*}
	To attain constant-factor approximation guarantees within 
	$O(n\log k)$ query complexity, we introduce a novel combinatorial algorithmic framework comprising two key phases:
	(1) In the first phase, we partition the solution space into two subspaces—one consisting of elements whose cardinality does not exceed $\alpha k$ for $\alpha \in (0, 1)$, and the other containing the remaining elements. We then compute near-optimal solutions independently within each subspace and strategically merge them to obtain a unified solution that achieves a constant approximation ratio of $0.044$ ( \algone \  Algorithm ).
	(2) In the second phase, we enhance the approximation ratio to 
	nearly $1/4$
	via the \algtwo \ algorithm. This improvement leverages the output of \algone \ and employs a carefully designed greedy thresholding technique to sequentially construct two candidate solutions. The construction is guided by the structural properties of two auxiliary vectors, 
	$\x$
	and 
	$\y$, which satisfy the condition 
	$\min \{\x(e),\y(e)\}=0$ for all 
	$e \in V$, as formalized in Lemma~\ref{lem:basic}.
	\paragraph{Organization}
	This paper is organized as follows: The Related Works section~\ref{sec:relatedwork}  presents previous research and existing methods. The Preliminaries section~\ref{sec:pre}  provides the basic concepts and theoretical foundations necessary for understanding the proposed algorithms. The  Proposed Algorithm section~\ref{sec:algs}  describes the details of the proposed algorithm, including the steps and underlying theory. The Experimental Evaluation section ~\ref{sec:expr}  presents the experimental results and analysis of the algorithm's performance. Finally, the Conclusion section ~\ref{sec:conclusion} summarizes the key findings and outlines directions for future research.
	
	\section{Related Works}
	\label{sec:relatedwork}
	In this section, we review the literature on optimization algorithms for $\DRSC$.
	\paragraph{Monotone DR-submodular functions.}  
	Soma and Yoshida~\cite{soma-nips-15} were the first to extend the notion of submodular functions to the integer lattice by introducing the concept of DR-submodularity and proposing a bicriteria approximation algorithm that simultaneously accounts for both the objective and cost functions. Subsequently, in a follow-up work~\cite{soma2018maximizing}, the authors developed two deterministic algorithms that achieve a \(1 - \frac{1}{e} - \epsilon\) approximation for maximizing monotone DR-submodular and general lattice submodular functions, with time complexities of \(O\left(\frac{n}{\epsilon}\log k \log \frac{k}{\epsilon}\right)\) and \(O\left(\frac{n}{\epsilon^2} \log k \log \frac{k}{\epsilon} \log \tau\right)\), respectively, where \(k\) denotes the size budget and the maximum coordinate bound, and \(\tau\) denotes both the ratio between the maximum function value and the minimum positive marginal gain.
	Lai et al. (2019)~\cite{lai2019monotone} introduced a randomized algorithm, though it suffers from instability in the number of oracle calls. More recently, Schiabel et al. (2025)~\cite{schiabel2021randomized} proposed the Stochastic Greedy Lattice (SGL) algorithm, which extends the stochastic greedy sampling strategy to monotone DR-submodular functions over integer lattices. SGL achieves an approximation guarantee of $\left(1 - \frac{1}{e} - \hat{t}\epsilon\right)$ with probability greater than $\frac{1}{2}$, where $\hat{t}$ is a small constant dependent on the number of iterations.
	
	\paragraph{Non-monotone DR-submodular functions.}  
	Designing algorithms for non-monotone DR-submodular functions is significantly more challenging than for the monotone case due to the potential decrease in function values. Gottschalk and Peis (2015)~\cite{gottschalk2015submodular} proposed a Double Greedy algorithm for maximizing submodular functions over bounded integer lattices, achieving a $\frac{1}{3}$ approximation ratio, albeit with pseudopolynomial complexity. Ene et al. (2020)~\cite{ene2020parallel} developed a parallel algorithm for non-monotone DR-submodular maximization but restricted it to the continuous domain. Li et al. (2023)~\cite{li2023dr} further improved algorithmic efficiency using an adaptive step size approach.
	
	As mentioned in the introduction, by applying the reduction technique proposed by Ene and Nguyen\cite{ene2016reduction}, the $\DSMS$ can be transformed into a standard submodular maximization problem under a knapsack constraint ($\SMK$). In the following, we review the line of research focused on the $\SMK$ problem. The first work to address the $\SMK$ problem in the non-monotone case was proposed by \cite{Lee_nonmono_matroid_knap}, achieving an approximation ratio of $5+\epsilon$ with polynomial query complexity. Subsequently, many follow-up studies have aimed to improve algorithmic efficiency~\cite{BuchbinderF19-bestfactor,Gupta_nonmono_constrained_submax,fast_icml,nearly-liner-der-2018,best-dla,pham-ijcai23,Han2021_knap}. Among these, the method in \cite{BuchbinderF19-bestfactor} achieved the best approximation factor of $2.6$, though at the cost of high oracle complexity. Conversely, the algorithm in \cite{pham-ijcai23} attained the lowest query complexity—linear in the input size—while achieving an approximation ratio of $4+\epsilon$.
	
	To the best of our knowledge, no existing work has directly addressed the problem of maximizing non-monotone DR-submodular functions under size constraints. In this paper, we propose a novel approach to address this problem.

	\section{Preliminaries}
	\label{sec:pre}
	For a positive integer $k \in \mathbb{N}$, $[k]$ denotes the set $\{1, \ldots, k\}$.
	Given a ground set  $E = \{e_1,\ldots, e_n\}$, we denote  $\x(e)$ as the value of $\x$’s coordinate
	corresponding to an element $e$, define the $e$-th unit vector $\1_e$ with $\1_e(t)=1$ if $t=e$ and $\1_e(t)=0$ if $t\neq e$. For $\x, \y \in \mathbb{Z}_+^E$, we say that $\x\leq \y$ iff $\x(e)\leq \y(e), \forall e \in E$.  For $\x=(x_1, x_2, \ldots, x_n)$ and $\y=(y_1, y_2, \ldots, y_n)$, we define
	\begin{align}
	\x \wedge \y&= (\min\{x_1, y_1\}, \min\{x_2, y_2\}, \ldots, \min\{x_n, y_n\})
	\\
	\x \vee \y&= (\max\{x_1, y_1\}, \max\{x_2, y_2\}, \ldots, \max\{x_n, y_n\}).
	\end{align}
	For function $f: \mathbb{Z}_+^E \mapsto \mathbb{R}_+$, we define $f(\x  |\y) =f(\x + \y)-f(\y)$. For $\x \in \mathbb{Z}^E_+$,
	we denote by $\{\x\}$ the set of elements  appears	in $\x$ times and with a subset, i.e, $\{\x\}=\{e \in E, \x(e)\geq 1 \}$ and the size of $\x$ as
	$\|\x\|_1=\sum_{e\in \{\x\}}\x(e)$. 
	A function $f: \mathbb{Z}_+^E \mapsto \mathbb{R}$ is \textbf{monotone}
	if $f(\x)\leq f(\y$) for all $\x,\y \in \mathbb{Z}_+^E$ with $\x \leq  \y $. The function $f$ is  \textbf{lattice submodular} iff
	\begin{align}
	f(\x)+ f(\y) \geq  f(\x \vee \y)+ f(\x \wedge \y)
	\end{align}
	for any $\x, \y \in \mathbb{Z}_+^E$.
	The function $f$ is said to be \textbf{DR-submodular} iff
	\begin{align}
	f(\x+ \1_e)-f(\x) \geq f(\y+\1_e)-f(\y)
	\end{align}
	for all $\x \leq \y$, $e \in E$. 
	If $f$ is DR-submodular then one  implies $f$ is lattice submodular \cite{soma-nips-15}.
	We assume that we can access $f$ through an oracle, i.e., for any vector $\x \in \mathbb{Z}_+^E$, it returns the value of $f(\x)$.
	
	The following basic lemmas are established as foundational tools for our theoretical analysis.
	\begin{lemma} For any $\s \in \mathbb{Z}^E_+$ and two vectors $\x, \y \in  \mathbb{Z}^E_+$ such that $\x \wedge \y=\0$ we have
		\begin{align}
		f(\s)\leq f(\s \vee\x ) +  f(\s \vee\y ).
		\end{align}
		\label{lem:basic}
	\end{lemma}
	\begin{proof}
		Since	$f$ is a non-negative and lattice submodular function, we have
		\begin{align}
		f(\s \vee\x ) +  f(\s \vee \y ) & \geq f(\s \vee \x \vee\y) +  f((\s \vee \x)\wedge (\s \vee \y) )  \label{lem1:ine1}
		\\
		& =f(\s \vee \x \vee\y) +  f(\s ) \label{lem1:ine2}
		\\
		& \geq f(\s) \label{lem1:ine3}
		\end{align}
		where the inequality~\eqref{lem1:ine1} is due to the lattice submodular, the inequality~\eqref{lem1:ine2} is due to $\x\wedge \y=\0$, the inequality~\eqref{lem1:ine3} is due to the non-negative of $f$. The proof is competed.
	\end{proof}
	\begin{lemma} For any $\x \in \mathbb{Z}^E_+$ and an integer number $t \geq 0$, then
		$
		f(t \1_e|\x) \leq t f(\1_e|\x)
		$.
		\label{lem:basic1}
	\end{lemma}
	\begin{proof}
		By the Dr-submodularity of $f$, we have
		\begin{align*}
		f(t\1_e \mid \x) &= f(\x+t\1_e) - f(\x)\\
		&=f(\1_e\mid \x+(t-1)\1_e)+f(\1_e\mid \x+(t-2)\1_e)+\ldots+f(\1_e\mid\x)\\
		&\leq tf(\1_e\mid \x) \label{lem:basic2}
		\end{align*}
		which completes the proof.
	\end{proof}
	\section{The Proposed Algorithms}
	\label{sec:algs}
	In this section, we present two algorithms: \textit{Fast Approximation   (\algone)}  and \textit{Fast  Approximation Plus  (\algtwo)}  for $\DSMS$ problem.  \algone\  retains an approximation ratio of $0.044$ (Theorem~\ref{theo:algone}) and takes the query complexity of $O(n \log k)$.  \algtwo \  built upon  \algone, enhances the optimization process by employing the constant number guess of optimal solutions with greedy threshold \cite{badanidiyuru2014fast} and therefore improve the approximation ratio to nearly  \cp{$\frac{1}{4}$} (Theorem~\ref{theo:algtwo}) without increasing the query complexity of $O(n \log k)$.
	
	\subsection{Fast Approximation Algorithm (\algone)}	
	
	The \algone\ algorithm takes as input a DR-submodular function $f$, a finite ground set $E$ of size $n$, a size budget $k \in \mathbb{Z}_{>0}$, and a threshold parameter $\alpha \in (0,1)$. Its goal is to output a solution vector $\z \in \mathbb{Z}_{\ge0}^{E}$ satisfying $\|\mathbf{z}\|_{1}\le k$, which approximately maximizes $f(\z)$. The core idea of the algorithm is to simultaneously construct two disjoint solution vectors, $\x$ and $\y$, using a dynamic-threshold rule, thus facilitating a clear theoretical analysis based on Lemma~\ref{lem:basic}.
	
	Specifically, it first performs a binary search over $d \in (\alpha k, k]$ to identify the pair $(d_{\max}, e_{\max})$ that maximizes $f(d\mathbf{1}_{e})$, retaining this singleton choice as an independent candidate (Line~\ref{d_max}). Next, starting from $\x = \y = \mathbf{0}$, the algorithm iterates through each element $e \in E$. For each element, it identifies the largest integer $d \leq \alpha k$ whose marginal gain satisfies $f(\mathbf{1}_e \mid \x + (d-1)\mathbf{1}_e) \geq \frac{f(\x)}{k}$ (and analogously for $\y$), then updates only the vector exhibiting the larger marginal gain. 
	
	Consequently, the two vectors evolve along distinct trajectories, guided by the same value-dependent threshold. After processing all elements, the algorithm trims the most recently added units from each vector, obtaining vectors $\x'$ and $\y'$ with $\|\x'\|_1 \le k$ and $\|\y'\|_1 \le k$. Finally, the best solution among $\x'$, $\y'$, and $d_{\max}\mathbf{1}_{e_{\max}}$ is returned. The overall procedure requires only $O(n\log k)$ oracle calls and guarantees a proven constant-factor approximation ratio.
	\begin{algorithm}[hpt]
		\label{algone}
		\KwIn{ $f: \mathbb{Z}^E_+ \mapsto \mathbb{R}_+$,$E$, $k$, $\alpha$ } 
		$(d_{max}, e_{max}) \leftarrow \arg \max_{e \in E, \ d \in \mathbb{Z}, \ \alpha k < d \leq k} \ f(d\1_{e}) $ by using the binary search \label{d_max}
		
		$\x \leftarrow \0$, $\y \leftarrow \0$ \label{setxy}
		
		\ForEach{$e \in E$}
		{\label{startloop}
			
			$d_{(\x,e)} \leftarrow  \max\{d \in \mathbb{Z} \mid 0 < d \leq \alpha k, \ f(\1_e \mid \x+(d-1)\1_e) \geq \frac{f(\x)}{k} \}$ by using the binary search \label{finddx}
			
			$d_{(\y,e)} \leftarrow  \max\{d \in \mathbb{Z} \mid 0 < d \leq \alpha k, \ f(\1_e \mid \y+(d-1)\1_e) \geq \frac{f(\y)}{k} \}$ by using the binary search \label{finddy}
			
			\eIf{$f(d_{(\x,e)} \1_e\mid \x) \geq f(d_{(\y,e)} \1_e\mid \y)$}{
				\label{comparexy}
				$\x \leftarrow \x + d_{(\x,e)}\1_e$\label{updatex} 
			}{
				$\y \leftarrow \y + d_{(\y,e)}\1_e$\label{updatey}
			}            
		}\label{endloop}
		Define $\{e_1, e_2, \ldots, e_t\}$ as the set of last $t$ elements added into $\x$, i.e, $\x_{t}=\sum_{i=1}^t \x(e_t) \1_{e_t}$
		\\
		$\x' \leftarrow \arg\max_{\x_t: |\x_t|\leq k}t$ \label{selectx}
		\\
		Define $\{e'_1, e'_2, \ldots, e'_t\}$ as the set of last $t$ elements added into $\y$, i.e, $\y_{t}=\sum_{i=1}^t \y(e'_t) \1_{e'_t}$
		\\
		$\y' \leftarrow \arg\max_{\y_t: |\y_t|\leq k}t$ \label{selecty}
		\\
		$\z \leftarrow \arg \max_{\z' \in \{\x', \y', d_{max}\1_{e_{max}}\}f(\z')}$ 
		\\
		\Return $\z$ \label{return}
		\caption{\algone \  Algorithm}
		\label{alg1}
	\end{algorithm}
	
	To analyze the theoretical guarantee of the Algorithm~\algone, we first provide following useful notations:
	\begin{itemize}    
		\item $\o$ is the optimal vector solution.
		\item  $\o_1$ is a sub-vector of $\o$ such that $
		\o_1(e) = 
		\begin{cases} 
		\o(e), & \text{if } \o(e) \leq \alpha k , \\
		0, & \text{otherwise}.
		\end{cases}
		$
		\item  $\o_2$ is a sub-vector of $\o$ such that $
		\o_2(e) = 
		\begin{cases} 
		\o(e), & \text{if } \o(e) > \alpha k , \\
		0, & \text{otherwise}.
		\end{cases}
		$
		\item  $E_\x = \{e \in E : \x(e)>0,\o(e)=0\}$, $E_{\x \wedge \o} = \{e \in E : \x(e)\geq  \o(e)>0\}$. 
		\item  $E_\y = \{e \in E : \y(e)>0,\o(e)=0\}$, $E_{\y\wedge \o} = \{e \in E : \y(e)\geq  \o(e)>0\}$.
		\item $E_{\o} = \{e \in E : \o(e)>0, \o(e)>\x(e),\o(e)>\y(e)\}$.
		\item  For $e \in E_\x\cup E_\y$, denote $t_e=\o_1(e) - \z(e), \z\in \{\x, \y\}$ (Note that $\x\wedge \y=\0$). 
		
		\item  $\x_e$ be $\x$  before $e$ added to $\x$ immediately at line \ref{updatex} for all $e\in \{\x\}$.
		\item $\y_e$ be $\y$  before $e$ added to $\y$ immediately at line \ref{updatey} for all $e\in \{\y\}$.
	\end{itemize}
We first establish the connection between $\x, \y$ and $\o_1$ as follows: 
	\begin{lemma}
		We have
		$f(\o_1 \vee \x)+f(\o_1 \vee \y) \leq 4(f(\x)+f(\y))$.
		\label{lem:xy}
	\end{lemma}
	\begin{proof}
   The main idea of the proof is to analyze the marginal gains when adding elements from $\o_1$. For $e \in E_{\o}$, the contribution can be bounded by a fraction of $\frac{f(\x)}{k}$, while for $e \in E_{\y \wedge \o}$, the gain obtained when adding to $\x$ is always upper-bounded by the gain when adding to $\y$ plus a small fraction proportional to $f(\x)$.
        \\
		Consider elements  $e \in E_\o$,
		due to the selection rule of $d_{(\x,e)}$ (lines \ref{finddx},\ref{finddy}) in each iteration,  we have $f(\1_e \mid \x_e+\x(e)\1_e)< \frac{f(\x_e)}{k}$.  Therefore 
		\begin{align}
		f(t_e\1_e\mid\x_e+\x(e)\1_e) &\leq t_ef(\1_e|\x_e+\x(e)\1_e) \ \ \ \mbox{(By Lemma~\ref{lem:basic1}})
		\\
		& \leq \frac{t_ef(\x_e)}{k}\leq\frac{\o_1(e)f(\x_e)}{k}.\label{proof_theo1_7}
		\end{align}
		Each element $e \in E_{\y\wedge \o}$ satisfied the selection rule in Lines 6-10. We consider two cases. If $d_{(\x, e)}\geq d_{(\y, e)}$, we have
		\begin{align}
		f(d_{(\y, e)} \1_e|\x) & \leq f(d_{(\y, e)} \1_e|\x) + \sum_{i=1}^{d_{(\x, e)}-d_{(\y, e)}} f(\1_e|\x+d_{(\y, e)} \1_e+ (i-1)\1_e)  \label{ine:a}
		\\
		&  = f(d_{(\x, e)} \1_e|\x)
		\\
		& \leq f(d_{(\y, e)} \1_e|\y)  \label{ine:b}
		\end{align}
		where the inequality~\eqref{ine:a} is due to the selection rule of $d_{(\x, e)}$, i.e,  $f(\1_e \mid \x+(d-1)\1_e) \geq \frac{f(\x)}{k} \geq 0$; the inequality~\eqref{ine:b} is due to the selection rule of $e$ and $d_{(\y, e)}$.
		\\
		If $d_{(\x, e)}<d_{(\y, e)}$, let $l_e=d_{(\y, e)}-d_{(\x, e)}$ we have
		\begin{align}
		f(d_{(\y, e)} \1_e|\x) &= f(d_{(\x, e)} \1_e|\x) + f(l_e \1_e|\x+ d_{(\x, e)} \1_e) 
		\\
		& \leq  f(d_{(\y, e)} \1_e|\y) + \frac{\y(e)f(\x)}{k} .
		\end{align}
		Put them together, we have
		\begin{align}
		f(\o_1 \vee \x)-f(\x)&\leq \sum_{e \in E_{\o_1}}f(t_e\1_{e}|\x)+\sum_{e \in E_{\y\wedge \o_1}}f(d_{(\y, e)}\1_{e}|\x)
		\\
		&\leq  \sum_{e \in E_{\o_1}}\frac{\o_1(e)f(\x_e)}{k}+\sum_{e \in E_{\y\wedge \o_1}}\Big(f(d_{(\y, e)} \1_e|\y) + \frac{\y(e)f(\x)}{k}\Big)
		\\
		&\leq  f(\x)+\sum_{e \in E_{\y\wedge \o_1}}\Big(f(d_{(\y, e)} \1_e|\y_e) + \frac{\y(e)f(\x)}{k}\Big)
		\\
		& \leq 2f(\x)+ f(\y). \label{x-ine}
		\end{align}
		By the similarity augments, we also have the same result as:
		\begin{align}
		f(\o_1 \vee \y)-f(\y)\leq   f(\x)+ 2f(\y).
		\end{align}
		Combining this with \eqref{x-ine} we have
		\begin{align}
		f(\o_1\vee\x)+f(\o_1\vee\y)\leq 4(f(\x)+f(\y))
		\end{align}
		which completes the proof.
	\end{proof}
	We state theoretical guarantees of \algone \ in Theorem~\ref{theo:algone}.
	\begin{theorem} For a constant $\alpha \in (0,1)$, Algorithm~\ref{alg1} provides an approximation ratio of $\frac{1}{8\frac{(2-\alpha)}{1-\alpha}+\frac{1}{\alpha}}$ and has a query complexity of $O(n \log(k))$. The algorithm achieves the ratio of $\frac{1}{17+4\sqrt{2}}\approx0.044$ when $\alpha=\frac{2\sqrt{2}-1}{7}$.
		\label{theo:algone}
	\end{theorem}
	\begin{proof}
            \cp{The proof is based on comparing the values of $f(\x')$ and $f(\y')$ with $f(\x)$ and $f(\y)$, where $\x'$ and $\y'$ are feasible solutions obtained after the truncation step. By exploiting the submodularity property and Lemma~\ref{lem:xy}, we bound $f(\o_1)$ by a constant factor of $f(\s)$, while $f(\o_2)$ is bounded using the selection rule of the element $(d_{\max}, e_{\max})$.}
            \\
		We first prove the approximation ratio. By the selection rule of the algorithm, $\x'$ is a feasible solution. If $\|\x\|_1 \leq k$, then $\x' = \x$. If $\|\x\|_1 > k$, we have
		\begin{align}
		f(\x'_t) - f(\x'_{t-1}) &=f(\x'(e_t)\1_e \mid \x'_{t-1}) 
		\\       
		& \geq \frac{\x'(e_t) f((\x-\x') \vee \x'_{t-1}) }{k}\\
		& \geq \frac{ \x'(e_t)f(\x-\x')}{k}.
		\end{align}
		By the rule of selection $\x'$ and $\x$ then $k- \alpha k\leq \|\x'\|_1 \leq k$, we get
		\begin{align}
		f(\x)-f(\x-\x')&=\sum_{i=2}^t(f(\x'_i) - f(\x'_{i-1}))\geq \sum_{i=2}^t  \frac{ \x'(e_t)f(\x-\x')}{k}
		\\
		& \geq \frac{\|\x'\|_1}{k} f(\x-\x') \geq (1-\alpha)f(\x-\x')
		\end{align}
		which implies that $f(\x-\x')\leq \frac{f(\x)}{2-\alpha}$.
		By the submodularity property, we have $f(\x)\leq f(\x')+f(\x-\x')$ implying that
		\begin{align}
		f(\x') &\geq  f(\x)- f(\x-\x') \geq \frac{1-\alpha}{2-\alpha}f(\x).
		\end{align}
		By the similar analysis for $f(\y)$, we  obtain: $f(\y')\geq \frac{f(\y)}{2}$. Next, by using Lemma~\ref{lem:xy}, we have:
		\begin{align}
		f(\o_1) &\leq f(\o_1 \vee \x) + f(\o_1 \vee \y)\\
		& \leq 4(f(\x) + f(\y))
		\\
		& \leq 4\frac{2-\alpha}{1-\alpha}(f(\x') + f(\y'))
		\\
		& \leq 8\frac{2-\alpha}{1-\alpha} f(\s).
		\end{align}
		On the other hand, from the definition of $\o_2$ and the selection rule of $(d_{max},e_{max})$ at line \ref{d_max}, we have:
		\begin{align}
		f(\o_2) &=f\left(\sum_{e\in \{\o_2\}}\o_2(e)\1_e\right)\leq \sum_{e\in \{\o_2\}}f(\o_2(e)\1_e) \label{ine:1}
		\\
		& \leq \frac{1}{\alpha}f(d_{max}\1_{e_{max}}) \leq \frac{f(\s)}{\alpha}
		\end{align}
		where the inequality in \eqref{ine:1} due to the DR-submodularity of $f$. 
		Put them together, we get
		\begin{align}
		f(\o)\leq f(\o_1)+f(\o_2) \leq \Big(8\frac{2-\alpha}{1-\alpha}+\frac{1}{\alpha}\Big) f(\s).
		\end{align}
     To obtain the best approximation ratio, we choose the parameter $\alpha$ that minimizes the expression $8\frac{2-\alpha}{1-\alpha}+\tfrac{1}{\alpha}$. This optimization yields $\alpha = \tfrac{2\sqrt{2}-1}{7}$, which leads to the final approximation ratio of $\approx0.044$. 
    We provide the detail explanations bellows:
        Consider the function
        \[
        \Phi(\alpha)\;=\;8\,\frac{2-\alpha}{1-\alpha}+\frac{1}{\alpha},\qquad \alpha\in(0,1),
        \]
        Compute the derivative:
        \[
        \Phi'(\alpha)
        =8\cdot\frac{d}{d\alpha}\!\left(\frac{2-\alpha}{1-\alpha}\right)-\frac{1}{\alpha^{2}}
        =8\cdot\frac{1}{(1-\alpha)^2}-\frac{1}{\alpha^{2}}.
        \]
        Setting $\Phi'(\alpha)=0$ gives
        \[
        \frac{8}{(1-\alpha)^2}=\frac{1}{\alpha^2}
        \;\;\Longleftrightarrow\;\;
        \frac{1-\alpha}{\alpha}=2\sqrt{2}
        \;\;\Longleftrightarrow\;\;
        \alpha^\star=\frac{1}{1+2\sqrt{2}}
        =\frac{2\sqrt{2}-1}{7}.
        \]
        The second derivative,
        \[
        \Phi''(\alpha)=\frac{16}{(1-\alpha)^3}+\frac{2}{\alpha^{3}}>0\quad\text{for all }\alpha\in(0,1),
        \]
  Therefore, evaluating $\Phi$ at $\alpha^\star$ yields the exact minimum
        \[
        \Phi(\alpha^\star)
        =8\,\frac{2-\alpha^\star}{1-\alpha^\star}+\frac{1}{\alpha^\star}
        =17+4\sqrt{2}
        \]
        Hence the resulting approximation ratio is the reciprocal,
        \[
        \frac{1}{\Phi(\alpha^\star)}=\frac{1}{\,17+4\sqrt{2}\,}\;\approx\;0.04414.
        \]
    
		We now prove the query complexity of the algorithm.
		To find $(e_{\text{max}}, d_{\text{max}})$, for each element $e$, we employ a binary search method within the interval $[0, \lceil \alpha k \rceil]$ to determine $d_e$. This task takes $n \log(\alpha k)$ queries.
		The algorithm consists of a main loop where each element $e$ is processed. For each element, we apply the binary search method to find $d_{(\x,e)}$ and $d_{(\y,e)}$. Consequently, the number of queries for each element $e$ is $2 \log(\alpha k)$.
		Thus, the total number of queries is
		$$n \log(k) + 2n \log(\alpha k)=O(n \log k).$$ The proof is completed.
	\end{proof}
Although this approximation ratio is relatively small, it comes with the significant advantage of very low query complexity $O(n\log k)$, making the method practically efficient. Moreover, Algorithm~\algone~ primarily serves as the foundation for Algorithm~\algtwo~by providing an approximation ratio guarantee and identifying an approximate range of the optimum; thus, it is sufficient for Algorithm~1 to achieve any constant approximation ratio with linear-time complexity.
	
	\subsection{Fast  Approximation Plus Algorithm (\algtwo)}
	
	The \algtwo\ algorithm takes as input a DR-submodular function $f(\cdot)$, a ground set $E$, and parameters $k \in \mathbb{Z}_+$, $\alpha \in (0,1)$ and $\epsilon \in (0,1)$, and returns a vector $\s \in \mathbb{Z}_+^E$ satisfying $\|\s\|_1 \leq k$, such that $f(\s)$ approximates the optimal value.
	
	The main idea of the algorithm is to utilize the preliminary solution $\s'$, obtained from Algorithm~\algone, to estimate an upper bound of the optimal value as $\Gamma$ (Line~\ref{setgamma}), and to initialize a selection threshold $\theta = \Gamma / (4k)$. 
	
	The algorithm then constructs three candidate vectors $\x$, $\y$, and $\z$ based on a decreasing-threshold strategy with a multiplicative decay factor of $(1 - \epsilon)$. In this process, $\x$ and $\y$ are built in an alternating manner to ensure that their supports remain disjoint, while $\z$ is constructed greedily and independently. In each iteration, for every element $e \in E$, the algorithm determines the largest number of units $d$ that can be added to each vector such that the marginal gain remains at least $\theta$, and updates the vectors accordingly. The loop continues until the threshold $\theta$ drops below $\epsilon \Gamma / (16k)$, at which point the algorithm evaluates the objective function $f$ over the four candidates $\s'$, $\x$, $\y$, and $\z$, and returns the one that yields the highest function value.
	
	\begin{algorithm}\label{algtwo}
		\SetNlSty{text}{}{:}
		\KwIn{ $f: \mathbb{Z}^E_+ \to \mathbb{R}_+, E, k, \alpha, \epsilon$ }
		
		$\s' \leftarrow \text{\algone}(f, E, k, \alpha)$\label{callAlg1}\\
		$\Gamma = f(\s')\left(8\frac{(2-\alpha)}{1-\alpha} + \frac{1}{\alpha} \right), \ \theta = \frac{\Gamma}{4k}$ \label{setgamma}\\
		$\x \leftarrow \0$, $\y \leftarrow \0$,$\z \leftarrow \0$  \label{alg2:setxyz}\\
		\While{$\theta \geq \frac{\epsilon \Gamma}{16k}$}
		{\label{alg2:mainloop}
			
			\ForEach{$e \in E$}
			{\label{alg2:startinloop}
				$d_\x \leftarrow \max\{d \in \mathbb{Z} \mid 0 < d \leq k - \|\x\|_1, \ f(\1_e \mid \x+(d-1)\1_e) \geq \theta \}$ by using the binary search \label{alg2:finddx}
				
				$d_\y \leftarrow \max\{d \in \mathbb{Z} \mid 0 < d \leq k - \|\y\|_1, \ f(\1_e \mid \y+(d-1)\1_e) \geq \theta \}$ by using the binary search \label{alg2:finddy}
				
				$d_\z \leftarrow \max\{d \in \mathbb{Z} \mid 0 < d \leq k - \|\z\|_1, \ f(\1_e \mid \z+(d-1)\1_e) \geq \theta \}$ by using the binary search \label{alg2:finddz}
				
				$\z \leftarrow \z+d_\z\1_e$ \label{alg2:updatez}
				
				\eIf{$f((d_\x+\x(e))\1_e \mid \x -\x(e)\1_e) \geq  f( (d_\y+\y(e))\1_e\mid \y - \y(e)\1_e)$}{\label{alg2:comparexy}
					$\x \leftarrow \x + d_\x \1_e$\label{alg2:setx}
					
					$\y \leftarrow \y-\y(e)\1_e$
				}{
					$\y \leftarrow \y + d_\y \1_e$\label{alg2:sety}
					
					$\x \leftarrow \x-\x(e)\1_e$
				}
				
			}
			$\theta \leftarrow (1 - \epsilon) \theta$ \label{alg2:updatetheta}
		}\label{alg2:endloop}
		$\s \leftarrow \arg \max_{\t \in \{\s', \x, \y, \z\}} f(\t)$\label{alg2:getsolution}
		\\
		\Return $\s$\label{alg2:return}
		\caption{Fast  Approximation Plus Algorithm (\algtwo)}
		\label{alg2}
	\end{algorithm}
	We analyze and provide the theoretical guarantees of \algtwo\ in Theorem~\ref{theo:algtwo}. We define the following notations for theoretical analysis.
	\begin{itemize}
		\item $\o$ is the optimal solution and $\opt=f(\o)$ 
		\item $\theta_e$ is the value of $\theta$ at the time  element $e$ is added to the vector $\x$ or $\y$.
		\item  $E_\x = \{e \in E : \x(e)>0,\o(e)=0\}$, $E_{\x \wedge \o} = \{e \in E : \x(e)\geq  \o(e)>0\}$. 
		\item  $E_\y = \{e \in E : \y(e)>0,\o(e)=0\}$, $E_{\y\wedge \o} = \{e \in E : \y(e)\geq  \o(e)>0\}$.
		\item $E_{\o} = \{e \in E : \o(e)>0, \o(e)>\x(e),\o(e)>\y(e)\}$
		\item $\x_e$   be $\x$  before $e$ is added to $\x$ immediately  for all $e\in \{\x\}$.
		\item  $\y_e$   be $\y$  before $e$ is added to $\x$ immediately  for all $e\in \{\y\}$.
	\end{itemize}
	We analyze and provide the theoretical guarantees of \algtwo\ in Theorem~\ref{theo:algtwo}.
	\begin{theorem} 
		For $\epsilon >0$ and the constant $\alpha \in (0, 1)$,  Algorithm~\ref{algtwo} provides an approximation ratio of $\frac{1}{4} -\epsilon$ and has a query complexity of $O(\frac{n}{\epsilon}\log(\frac{1}{\epsilon})\log(k))$.
		\label{theo:algtwo}
	\end{theorem}
	\begin{proof}
          The main idea of the proof is to bound the marginal contributions of elements added to the solution $\x,\y$. By the selection rule and threshold updates, each chosen element contributes at least as much as the threshold in its iteration, ensuring that the total value is not much smaller than $\opt$. This leads to the $(\tfrac{1}{4}-\epsilon)$ approximation guarantee, with query complexity obtained from the use of binary search in the loops.
            \\
		After the first completion of the $while$ loop, we have two possible cases: 
		\\
		(1) If $\|\x\|_1 = k$ or $\|\y\|_1 = k$. Assume $\|\x\|_1 = k$ by the selection rule of $e \in \{\x\}$ we have $$f(\x)\geq \|\x \|_1 \frac{\Gamma}{4k} \geq \frac{\opt}{4}.$$
		Therefore, the approximation ratio holds. It is similar to the case $\|\y\|_1 = k$.
		\\
		(2)	If	 $\|\x\|_1 < k$ and $\|\y\|_1 < k$.
		Denote $\theta'$ as $\theta$ when the last time an element $e$ to the vector $\x$. Considering the vector $\x$ at the end of the while loop, we examine two subcases:
		\begin{itemize}
			\item If $\|\x\|_1 = k$.
			Assume that $\theta_{(l)}$ is $\theta$ when the last element is added into $\x$. We obtain
			\begin{align}
			f(\x)=\sum_{e\in \{x\}}f(\x(e)\1_e|\x_e)\geq k\theta_{(l)}.
			\end{align}
			Any $e \in E_{\o}$ not added to $\x$ has a marginal gain per less than the $\theta$  in the previous iteration, i.e, $f(\1_e|\x_e)< \frac{\theta_{(l)}}{1-\epsilon}$. Since $\|\x\|_1=k\geq \|\o\|_1$ we have 
			\begin{align}
			\sum_{e\in E_{\x}}\x(e) + \sum_{e\in E_{\x \wedge\o}}\x(e) + \sum_{e\in E_{\o}}\x(e)=k\geq  \sum_{e\in E_{\o}}\x(e) + \sum_{e\in E_{\x \wedge\o}}\o(e) + \sum_{e\in E_{\o}}t_e
			\end{align}
			which implies that $\sum_{e\in E_{\x}}\x(e)\geq   \sum_{e\in E_{\o}}t_e$. Put them together, we have
			\begin{align}
			\sum_{e\in E_{\o}}f(t_e\1_e|\x)&\leq  \sum_{_e\in E_{\o}}t_e f(\1_e|\x_e) \ \ \mbox{(By Lemma~\ref{lem:basic1})}
			\\
			& \leq \sum_{e\in E_{\o}} \frac{t_e\theta_{(l)}}{1-\epsilon}=\frac{\theta_{(l)}}{1-\epsilon} \sum_{e\in E_{\o}} t_e
			\\
			&  \leq  \frac{\theta_{(l)}}{1-\epsilon} \sum_{e\in E_{\x}} \x(e)= \sum_{e\in E_{\x}}  \frac{\x(e)\theta_{(l)}}{1-\epsilon}
			\\
			& \leq   \sum_{e\in E_{\x}} \frac{\x(e)\theta_{e}}{1-\epsilon}
			\\
			&  \leq \sum_{e\in E_{\x}} \sum_{d=1}^{\x(e)}\frac{f(\1_e|\x_e+(d-1)\1_e)}{1-\epsilon}
			\\
			&  \leq \sum_{e\in E_{\x}}\frac{f(\x(e)\1_e|\x_e)}{1-\epsilon}.
			\end{align}
			\item If $\|\x\|_1 < k$. In this case any $e \in E_{\o}$ not added to $\x$ has a marginal gain per less than the $\theta$  at the last iteration. So we have $f(\1_e|\x_e)< \epsilon \frac{M}{k}$ and thus
			\begin{align}
			\sum_{e\in E_{\o}}f(t_e\1_e|\x)
			& \leq \sum_{e\in E_{\o}} t_e   \frac{\epsilon M}{k}
			\leq  \epsilon M \leq \epsilon\opt.
			\end{align}
			Combining two cases, we have
			\begin{align}
			\sum_{e\in E_{\o}}f(t_e\1_e|\x)
			& \leq \epsilon\opt +\sum_{e\in E_{\x}}\frac{f(\x(e)\1_e|\x_e)}{1-\epsilon}.	\label{pro:x}
			\end{align}
			Similarity, we  have the same property for $\y$
			\begin{align}
			\sum_{e\in E_{\o}}f(t_e\1_e|\y)
			& \leq \epsilon\opt +\sum_{e\in E_{\y}}\frac{f(\y(e)\1_e|\y_e)}{1-\epsilon}
			\label{pro:y}.
			\end{align}
		\end{itemize}
		For  $e \in E_{\o\wedge \y}$, assume $d_{(\y, e)} \1_e$ is added $\y$ at the iteration $j$, it was not added $\y$ at the previous iteration,  we have $f(\1_e|\x)\leq\frac{\theta_e}{1-\epsilon}$. Thus
		\begin{align}
		f(\y(e) \1_e|\x)=	f(d_{(\y, e)} \1_e|\x) < \frac{d_{(\y, e)} \theta_e}{1-\epsilon}\leq 	\frac{f(d_{(\y, e)} \1_e|\y_e)}{1-\epsilon}=\frac{f(\y(e)\1_e|\y_e)}{1-\epsilon}.
		\end{align}
		Put them together, we analyze the relationship between $f(\o \vee \x)$ and $f(\x)$ as follows:
		\begin{align}
		f(\o \vee \x) - f(\x) &\leq \sum_{e \in E_\o} f(t_e \1_e \mid \x) + \sum_{e \in E_{\o\wedge \y}} f(\y(e) \1_e \mid \x) 
		\\
		&\leq \epsilon\opt +\sum_{e\in E_{\x}}\frac{f(\x(e)\1_e|\x_e)}{1-\epsilon} +  \sum_{e \in E_{\o\wedge \y}} f(\y(e) \1_e \mid \y_e)\label{proof:x} 
		\end{align}
		Similarly, we  have the same property for $\y$
		\begin{align}
		f(\o \vee \y) - f(\y) 
		&\leq \epsilon\opt +\sum_{e\in E_{\y}}\frac{f(\y(e)\1_e|\x_e)}{1-\epsilon} +  \sum_{e \in E_{\o\wedge \x}} f(\x(e) \1_e \mid \x_e)\label{proof:y} 
		\end{align}	
		Combining \eqref{proof:x} and \eqref{proof:y}, we have
		\begin{align}
		f(\o)-f(\x)-f(\y)&\leq 	f(\o \vee \x) - f(\x)+	f(\o \vee \y) - f(\y) \\
		& \leq \sum_{e\in E_{\x}}\frac{f(\x(e)\1_e|\x_e)}{1-\epsilon} +     \sum_{e \in E_{\o\wedge\x}} f(\x(e) \1_e \mid \x_e) 
		\\
		&+\sum_{e\in E_{\y}}\frac{f(\y(e)\1_e|\x_e)}{1-\epsilon} + \sum_{e \in E_{\o\wedge \y}} f(\y(e) \1_e \mid \y_e)
		\\
		& \leq \frac{f(\x)+f(\y)}{1-\epsilon}.
		\end{align}
		Therefore
		\begin{align}
		f(\s) \geq \frac{f(\x)+f(\y)}{2} \geq \frac{1-\epsilon}{4-2\epsilon}f(\o)=(\frac{1}{4}-\frac{\epsilon}{2(4-2\epsilon)})\opt \geq (\frac{1}{4}-\epsilon)\opt.
		\end{align}
		For the query complexity, the \textbf{while} loop has a query complexity of \cp{$O(\frac{1}{\epsilon}\log(\frac{1}{ \epsilon}))$}. The nested \textbf{for} loop has a query complexity of $O(n \log(k))$. Therefore, the overall query complexity of the algorithm is \cp{$O(\frac{n}{\epsilon}\log(\frac{1}{ \epsilon})\log(k))$}.
		The proof is completed.
	\end{proof}

    Oveis Gharan and Vondrák~\cite{gharan2011submodular} established that no algorithm can achieve an approximation ratio better than $0.478$ for general DR-submodular maximization, which sets a natural upper bound on what can be expected. Within this landscape, the best-known guarantee is $0.401$ \cite{buchbinder2024constrained}; however, the underlying algorithm requires polynomial complexity and is therefore impractical for large-scale applications. On the other hand, the fastest near-linear time algorithms achieve an approximation ratio of $\tfrac{1}{4}$ \cite{pham-ijcai23,Han2021_knap}, but these methods fundamentally rely on randomness. Our proposed algorithm \algtwo\ closes this gap: it attains the same $\tfrac{1}{4}$ approximation ratio with comparable near-linear complexity, while remaining entirely \emph{deterministic}. To the best of our knowledge, this makes \algtwo\ the first deterministic algorithm to simultaneously combine near-linear running time with the strongest known approximation guarantee in this problem.

	\section{Experimental Evaluation}
	\label{sec:expr}
	In this section, we present our experimental results and demonstrate the superiority of our algorithm compared to other baseline algorithms.
	\subsection{Experiment Settings}
	
	\subsubsection{Compared Algorithms}
	As discussed in the related work, no existing research addresses the Non-Monotone DR-submodular Maximization under the Size Constraint problem. Therefore, we evaluate our proposed algorithms by comparing them against the state-of-the-art methods for submodular function maximization under a knapsack constraint, specifically RLA and SMKRANACC. The comparison is detailed as follows:
	
	\begin{itemize}
		\item \textbf{\algone \ (this work):} Approximation ratio of $0.044$  with a time complexity of $ O(n \log(k)) $.
		\item \textbf{\algtwo  \ (this work):} Approximation ratio of $ \frac{1}{4} - \epsilon $ with a time complexity of $ O\left(\log(\frac{1}{\epsilon})\frac{n}{\epsilon} \log\left(\frac{k}{\epsilon}\right)\right)$.
		\item \textbf{RLA~\cite{pham-ijcai23}:} Approximation ratio of $\frac{1}{4}-\epsilon$ with a time complexity of $O\left(\frac{n}{\epsilon} \log (k) \log \frac{1}{\epsilon} \right)$.
		\item \textbf{SMKRANACC~\cite{Han2021_knap}:} Approximation ratio of $ \frac{1}{4} - \epsilon $ with a time complexity of $O\left( \frac{n}{\epsilon} \log (k) \log \frac{k}{\epsilon} \right)$.
	\end{itemize}
	
	To enable the execution of the RLA and SMKRANACC algorithms in the DR-submodular setting, we applied the necessary modifications based on the reduction framework proposed by Ene and Nguyen~\cite{ene2016reduction}.
	
	\subsubsection{Applications and Datasets}
	
	Based on the Revenue Maximization application using the Submodular function, we build the Revenue Maximization application using DR-submodular function as follows: Consider an undirected social network represented as a graph $ G = (V, E) $, where $ V $ is the set of users (nodes) and $ E $ is the set of connections (edges) between them. Each edge $ (u, v) \in E $ is associated with a non-negative weight $ w_{uv} \in [0, 1] $ representing the strength of influence from user $ u $ to user $ v $. We assume a limited advertising budget that needs to be allocated to maximize product adoption in the network. For each user $ u $, an investment $ \x(u) $ represents the level of effort to promote the product through that user.
	
	The set of users receiving investment is $ \{\x\} = \{ u \in V \mid \x(u) > 0 \} $. Our objective is to maximize the expected revenue, which is the number of users who will adopt the product based on the influence of users in $ \{\x\} $ on the rest of the network.
	
	To quantify the expected revenue, we define the revenue function $ f(\x) $ as follows:
	
	\[
	f(\x) = \sum_{u \in V \setminus \{\x\}} f_u ( \sum_{v \in \{\x\}} w_{uv} \x(v))
	\]
	
	where:
	\begin{itemize}
		\item $ f_u(t) = \log(1 + t^{\alpha_u}) $ is a concave, non-negative function representing the likelihood that user $ u $ will adopt the product given the cumulative influence $ t = \sum_{v \in \{x\}} w_{uv} \x(v) $. Here, $ \alpha_u \in (0, 1) $ is a personalized parameter for each user, modeling the saturation effect of influence from $ \{\x\} $ to $ u $.
		\item  $ w_{uv} $ represents the influence strength from node $ u $ to node $ v $.
	\end{itemize}
	
	The revenue function $f(\x)$ is DR-submodular function. This can be proven by considering the concavity of $ f_u(t) $, as $ f_u(t) $ is a concave function, meaning that the increment in $ f_u(t) $ decreases as $ t $ increases. When adding an element to the set of invested users $ \{\x\} $, the effect on the revenue will be stronger for vectors $ \x $ with fewer elements (i.e., users who haven't received investment). Thus, adding an element results in a stronger increase in revenue for vectors with fewer elements, demonstrating the diminishing return on influence, and satisfying the DR-submodular inequality:
	
	\[
	f(\x + \1_e) - f(\x) \geq f(\y + \1_e) - f(\y)
	\]
	
	for all $ \x \leq \y $. Therefore, the revenue function $ f(\x) $ is DR-submodular.
	
	We utilize three datasets (Facebook, AstroPh, and Enron) to evaluate the performance of our algorithms. These datasets vary in size: Facebook is the smallest, containing only $4,039$ vertices, whereas Enron is the largest, comprising $36,692$ vertices. AstroPh, with $18,772$ vertices, falls between the two in terms of size. These benchmark datasets are obtained from \href{https://snap.stanford.edu/data/email-Enron.html}{SNAP} (see Table~\ref{tab:dataset}).
	
	\paragraph{Other settings.} We executed the \algone\ algorithm with various values of $\alpha$ in the set $\{0.1, 0.3, 0.5, 0.7, 0.9\}$, denoted in the charts as \algone-$0.1$, $\ldots$, \algone-$0.9$. For the \algtwo\ algorithm, we fixed the parameters as $\alpha = \frac{2\sqrt{2}-1}{7} \approx 0.2612$. All algorithms were executed with the parameter $\epsilon = 0.1$. For the randomized algorithms RLA and SMKRANACC, we performed 10 runs and reported the average result. The experiments were conducted on an HPC server cluster with the following specifications: partition = large, number of CPU threads = 16, number of nodes = 2, and maximum memory = 3,073 GB. The running time of each algorithm was measured in seconds and includes both the oracle queries and the additional computational overhead incurred by the respective procedures.  In all result figures, the x-axis representing the budget $k$ is normalized by $n$, the size of the corresponding dataset.

	\begin{table}[hpt]
		\caption{Details of datasets used in the experiments.}
		\label{tab:dataset}
		\centering
		\footnotesize{
			\begin{tabular}{ccccc}
				\hline
				\textbf{Dataset} & \textbf{Nodes} & \textbf{Edges} & \textbf{Types} & \textbf{Sources}
				\\
				\hline
				Facebook & 4,039 & 88,234 & Undirected  & \href{https://snap.stanford.edu/data/ego-Facebook.html}{SNAP}
				\\
				AstroPh& 18,772 & 198,110 & Undirected  & \href{https://snap.stanford.edu/data/ca-AstroPh.html}{SNAP}
				\\
				Enron & 36,692 & 183,831 & Undirected  & \href{https://snap.stanford.edu/data/email-Enron.html}{SNAP}					
				\\
				\hline
			\end{tabular}
		}
	\end{table}
	
	\subsection{Experiment Results}
	
	\begin{figure}
		\centering
		\includegraphics[width=0.95\linewidth]{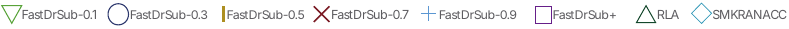}
		\\
		\includegraphics[width=0.32\linewidth]{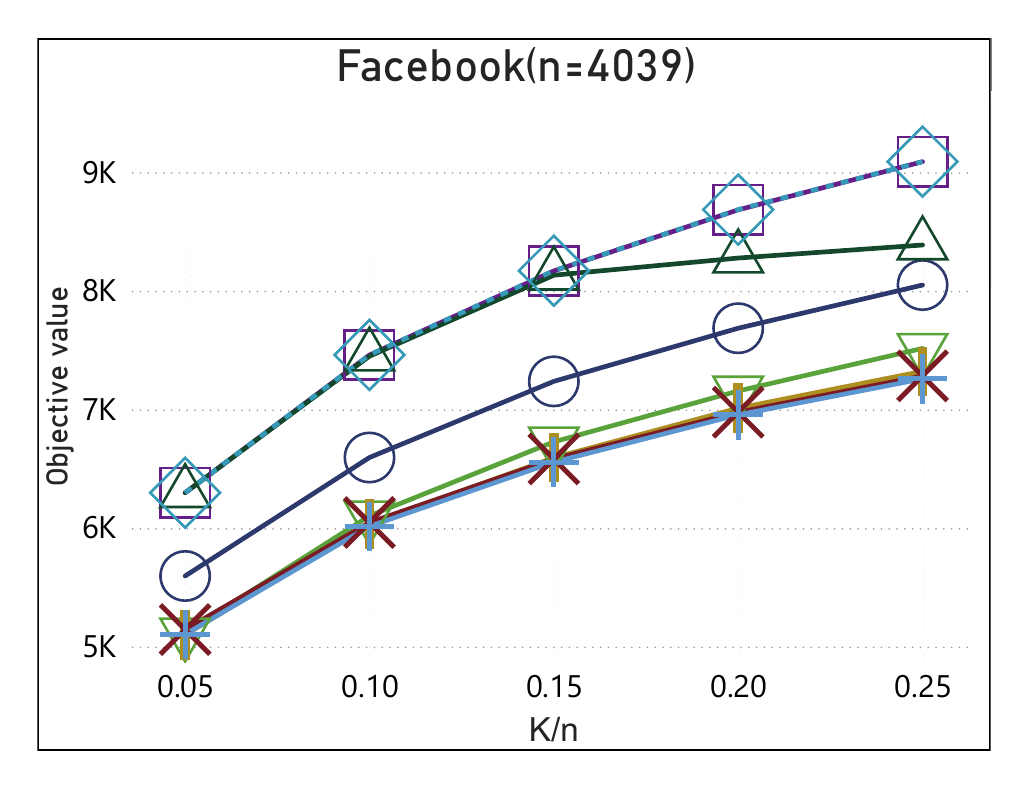}
		\includegraphics[width=0.32\linewidth]{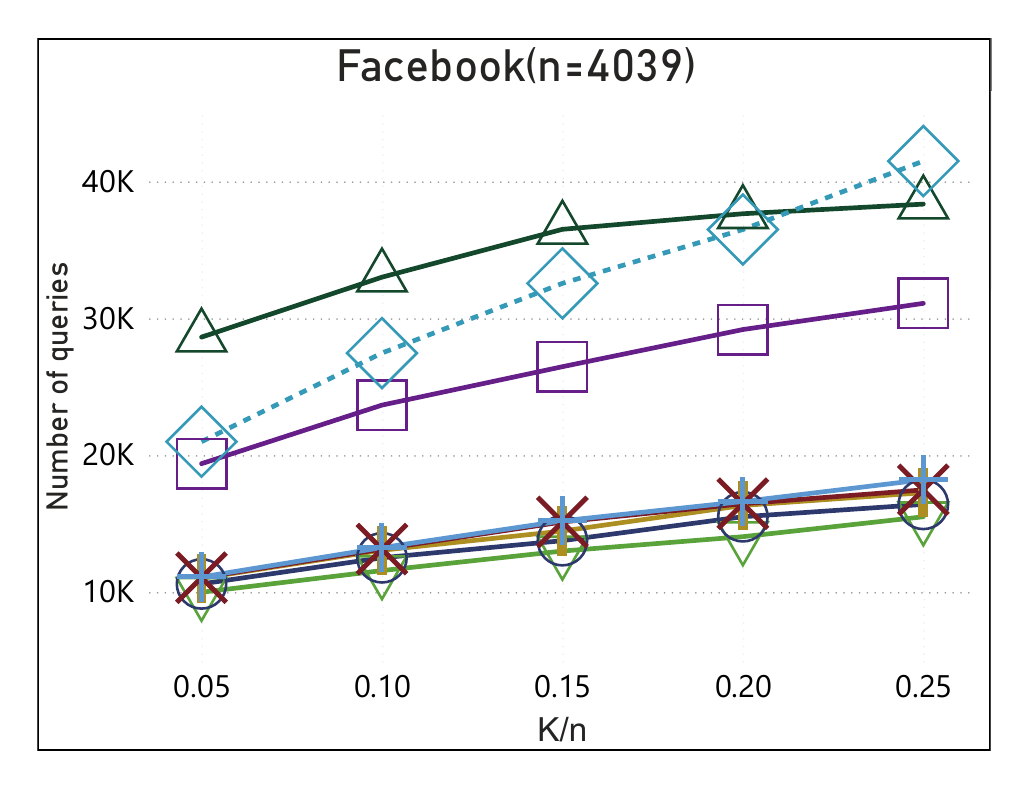}
		\includegraphics[width=0.32\linewidth]{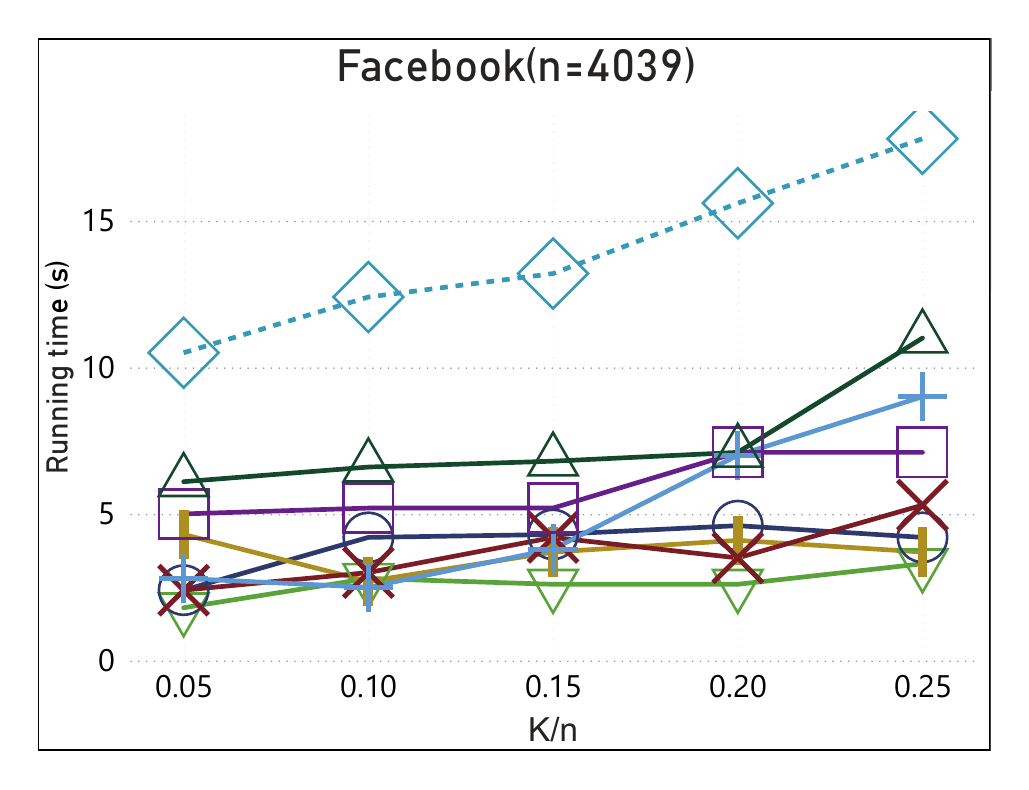}
		\\
		(a) \hspace{3.5cm} (b) \hspace{3.5cm} (c)
		\caption{Performance of algorithms on Revenue Maximization for the Facebook dataset: (a) The objective values, (b) The number of queries and (c) their running time.}
		\label{fig:facebook}
	\end{figure}
	
	\begin{figure}
		\centering
		\includegraphics[width=0.95\linewidth]{legend.pdf}
		\\
		\includegraphics[width=0.32\linewidth]{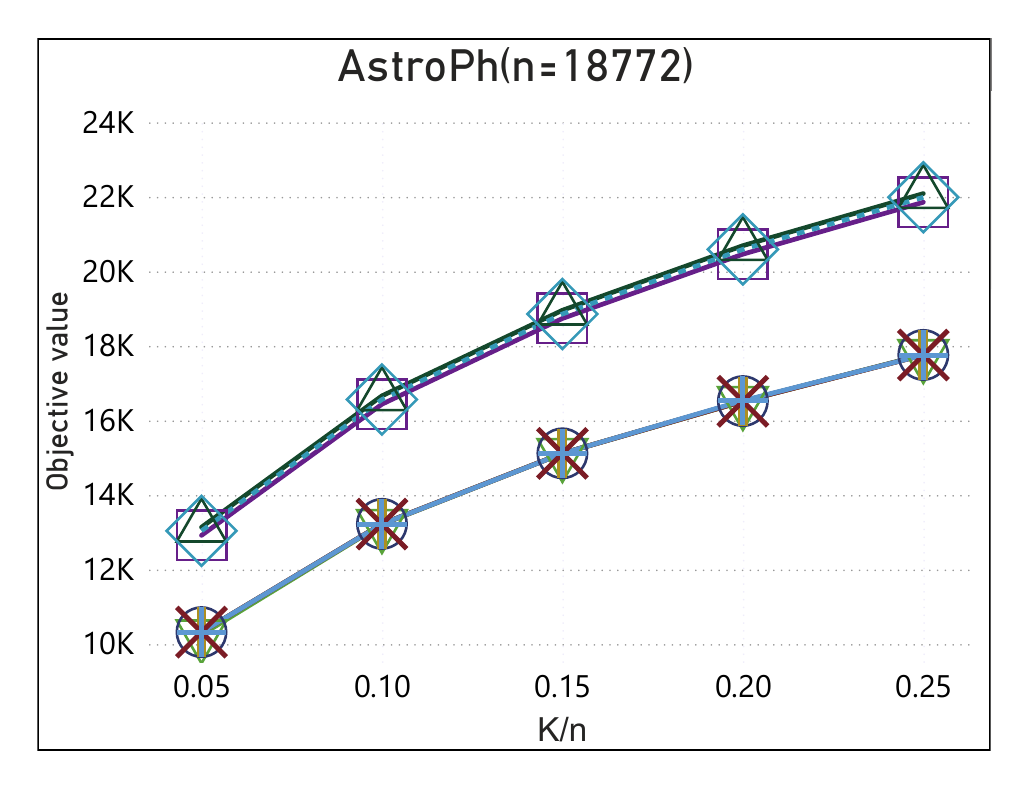}
		\includegraphics[width=0.32\linewidth]{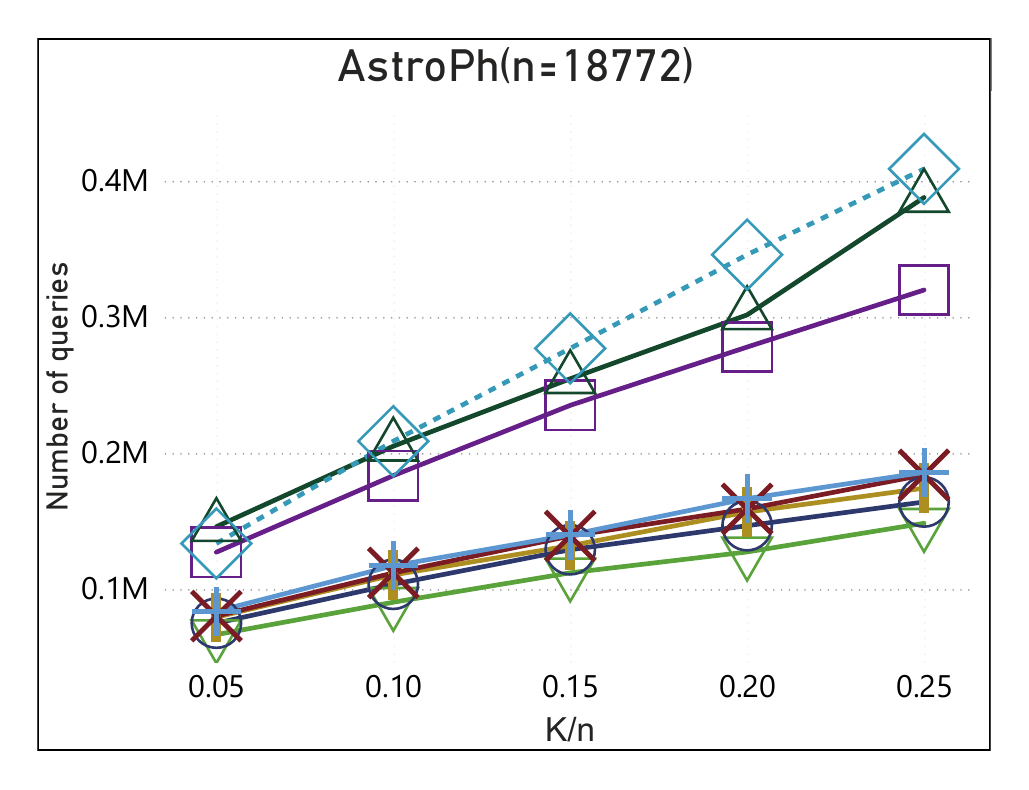}
		\includegraphics[width=0.32\linewidth]{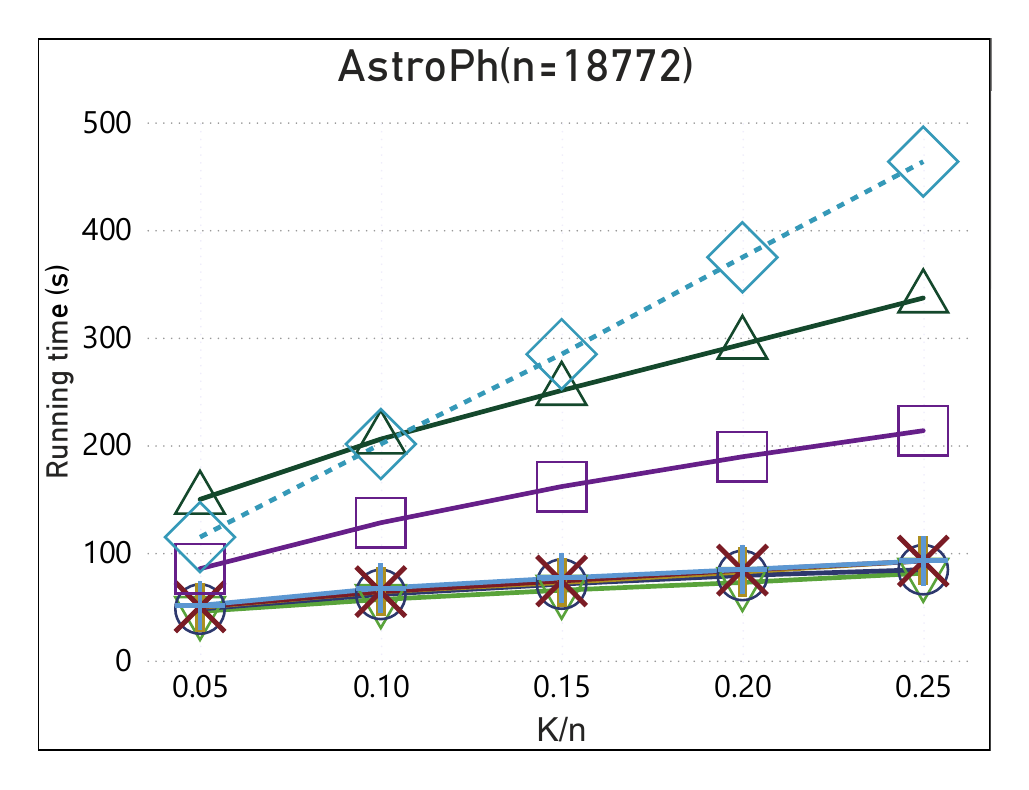}
		\\
		(a) \hspace{3.5cm} (b) \hspace{3.5cm} (c)
		\caption{Performance of algorithms on Revenue Maximization for the AstroPh dataset: (a) The objective values, (b) The number of queries and (c) their running time.}
		\label{fig:astro}
	\end{figure}
	
	\begin{figure}
		\centering
		\includegraphics[width=0.95\linewidth]{legend.pdf}
		\\
		\includegraphics[width=0.32\linewidth]{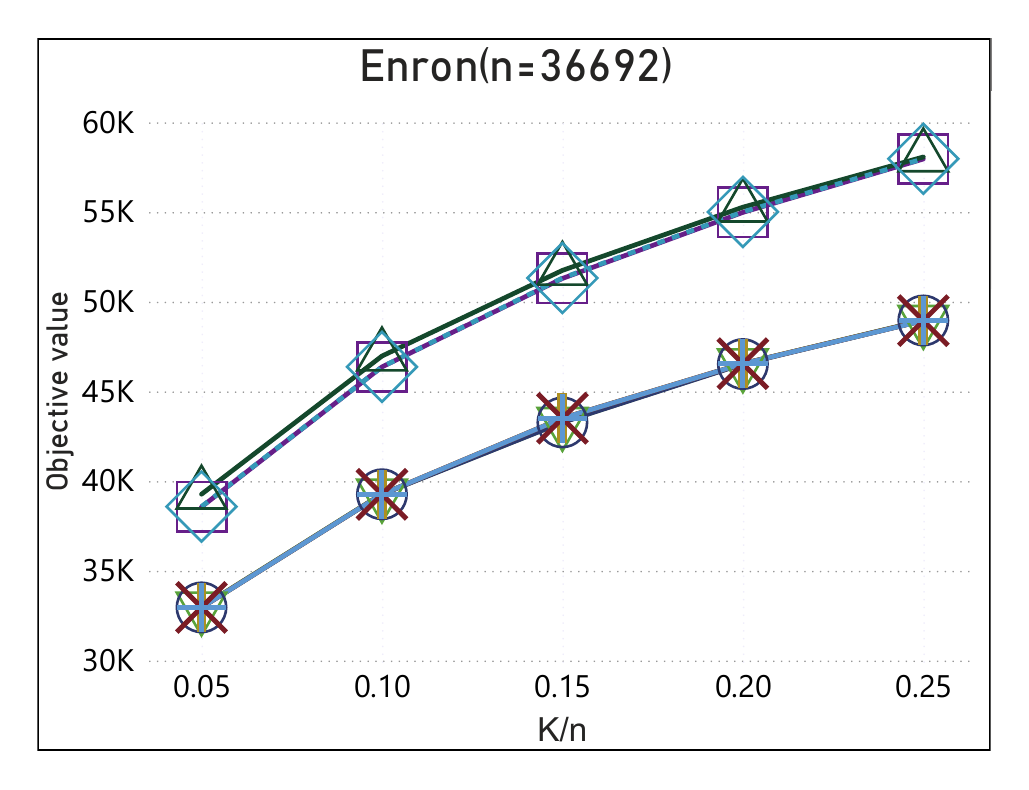}
		\includegraphics[width=0.32\linewidth]{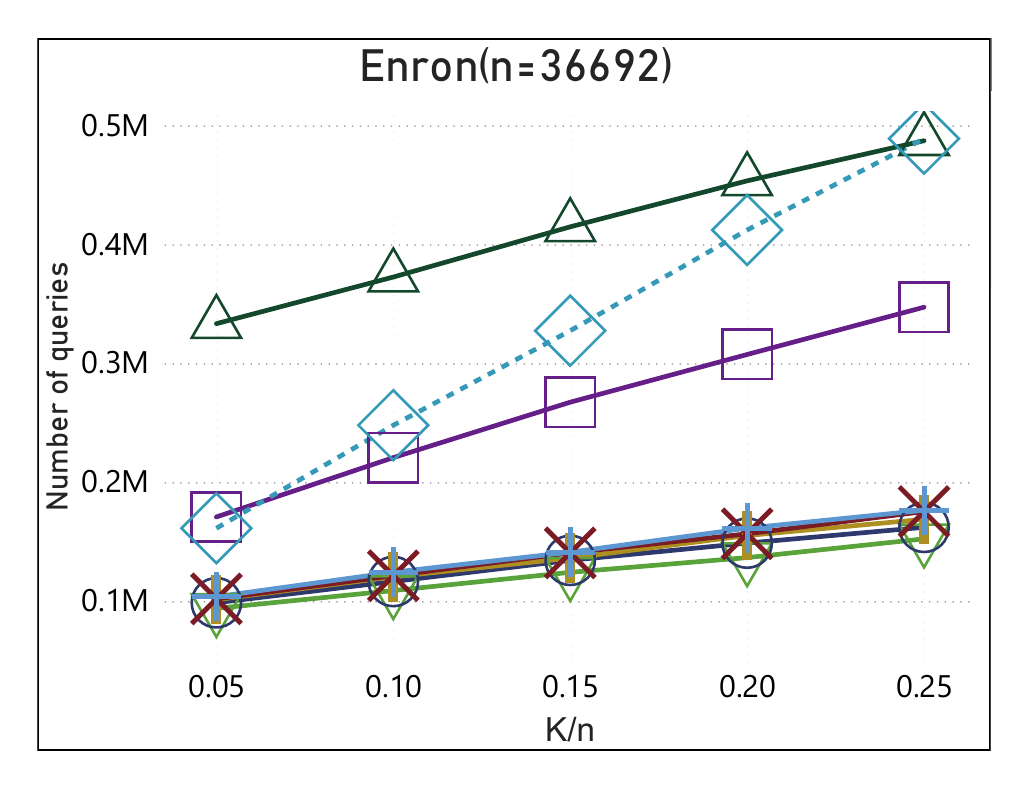}
		\includegraphics[width=0.32\linewidth]{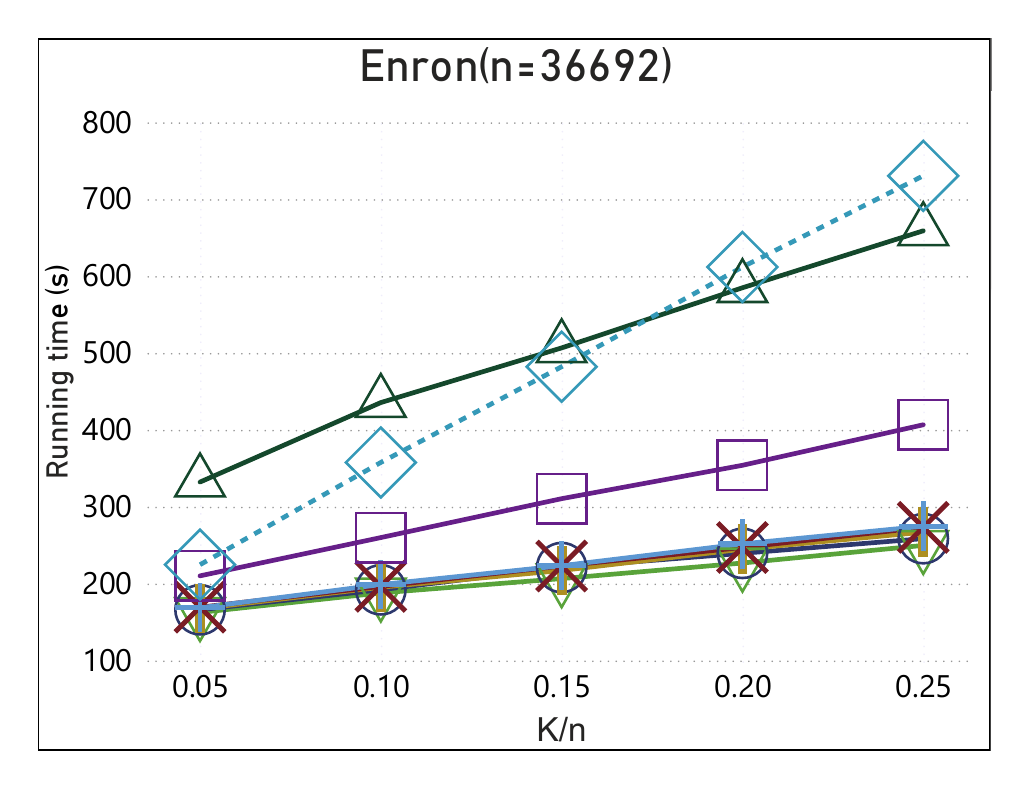}
		\\
		(a) \hspace{3.5cm} (b) \hspace{3.5cm} (c)
		\caption{Performance of algorithms on Revenue Maximization for the Enron dataset: (a) The objective values, (b) The number of queries and (c) their running time.}
		\label{fig:enron}
	\end{figure}
	The results for the objective function values are presented in Figures~\ref{fig:facebook}a, \ref{fig:astro}a, and \ref{fig:enron}a. Among all methods, \algtwo\ consistently achieves high objective function values, comparable to RLA and SMKRANACC, and outperforms the baseline \algone\ variants with margins ranging from 1.2 to 1.4 times. Across all three datasets, \algtwo\ maintains stable and competitive performance, with observed differences between high-performing algorithms being statistically insignificant. In the Facebook dataset, RLA shows slightly lower performance than \algtwo\ and SMKRANACC at the 0.2 and 0.25 marks. For the \algone\ algorithm with varying $\alpha$ values, the objective values are relatively similar in the AstroPh and Enron datasets due to their large size. In contrast, for the Facebook dataset, $\alpha = 0.3$ yields the highest outcome among \algone\ settings, while $\alpha = 0.9$ gives the lowest, although the variation remains modest.
	
	The results for the number of queries are shown in Figures~\ref{fig:facebook}b, \ref{fig:astro}b, and \ref{fig:enron}b. The algorithm \algtwo\ requires more queries than \algone\ variants, yet it generates significantly fewer queries than RLA and SMKRANACC. This efficiency is particularly evident in the Facebook and Enron datasets, where \algtwo\ produces approximately 1.2 times fewer queries than RLA, indicating its favorable balance between performance and computational cost. The number of queries for \algone\ remains consistent across different $\alpha$ values, with only slight variations across datasets.
	
	Figures~\ref{fig:facebook}c, \ref{fig:astro}c, and \ref{fig:enron}c present the results for running time. The running time trends closely mirror the query counts, as expected. The \algone\ algorithm with various $\alpha$ values achieves the shortest running times across all datasets, with minimal differences among $\alpha$ configurations. However, \algtwo\ demonstrates a strong trade-off between solution quality and efficiency—it is faster than both RLA and SMKRANACC while delivering comparable objective values. In particular, on the AstroPh and Enron datasets, \algtwo\ runs 1.5 to 2 times faster than RLA and SMKRANACC at $k/n = 0.05$ and $k/n = 0.25$, which demonstrates its practical efficiency.
	\section{Conclusion}
	This paper addresses the DrSMS problem, which generalizes the classical submodular maximization problem by incorporating diminishing returns over integer lattices. We introduced two efficient approximation algorithms designed to tackle this problem, with Algorithm \algone\ providing an approximation ratio of $0.044$ and Algorithm \algtwo\ achieving an even better ratio of $\frac{1}{4} -\epsilon$. 
	Our theoretical analysis provides strong guarantees on both algorithms' approximation quality and query complexity, demonstrating their effectiveness in various scenarios. 
	
	Furthermore, through extensive experimental evaluation, we showed that our algorithms outperform existing state-of-the-art methods, particularly in the context of Revenue Maximization applications, which leverage DR-submodular functions. The results highlight the practical value and efficiency of the proposed algorithms, making them suitable for large-scale problems where computational resources are limited.
	
	Future work could further improve these algorithms' scalability, explore additional problem formulations, and apply them to real-world applications with more complex constraints and objectives.
	
	\label{sec:conclusion}
	
	\section*{Acknowledgement}
	The first author (Tan D. Tran) was funded by the Master, PhD Scholarship Programme of Vingroup Innovation Foundation (VINIF), code VINIF.2024.TS.069. This work has been carried out partly at the Vietnam Institute for Advanced Study in
	Mathematics (VIASM). The second author (Canh V. Pham)
	would like to thank VIASM for its hospitality and financial
	support.

	%
	%

	\bibliographystyle{spmpsci}      
	\bibliography{drsub-ref}   
	
	
\end{document}